\documentclass{llncs}

\pdfoutput=1

\usepackage{amsfonts}
\usepackage{amsmath}
\usepackage{amssymb}
\usepackage{color}
\usepackage{graphicx}
\PassOptionsToPackage{hyphens}{url}\usepackage[pdftitle={HalftimeHash: Modern Hashing without 64-bit Multipliers or Finite Fields},hidelinks]{hyperref}
\usepackage{microtype}

\DeclareMathOperator{\adj}{adj}

\newcommand{\ints}{\mathbb{Z}}

\newcommand{\defeq}{\;\genfrac{}{}{0pt}{2}{\text{def}}{=}\;}

\begin{document}

\title{HalftimeHash: Modern Hashing without 64-bit Multipliers or Finite Fields}
\author{Jim Apple
\orcidID{0000-0002-8685-9451}}
\institute{  \email{\href{mailto:jbapple@apache.org}{jbapple@apache.org}}}
\maketitle

\begin{abstract}
HalftimeHash is a new algorithm for hashing long strings.
The goals are few collisions (different inputs that produce identical output hash values) and high performance.

Compared to the fastest universal hash functions on long strings (clhash and UMASH) HalftimeHash decreases collision probability while also increasing performance by over 50\%, exceeding 16 bytes per cycle.

In addition, HalftimeHash does not use any widening 64-bit multiplications or any finite field arithmetic that could limit its portability.

\end{abstract}

\section{Introduction}
A hash family is a map from a set of seeds $S$ and a domain $D$ to a codomain $C$.
A hash family $H$ is called is $\varepsilon$-almost universal (``$\varepsilon$-AU'' or just ``AU'') when
\[
\forall x,y \in D, x \neq y \implies \mathrm{Pr}_{s \in S}[H(s, x) = H(s, y)] \leq \varepsilon \in o(1)
\]
The intuition behind this definition is that collisions can be made unlikely by picking randomly from a hash {\em family} independent of the input strings, rather than anchoring on a specific hash {\em function} such as MD5 that does not take a seed as an input. AU hash families are useful in hash tables, where collisions slow down operations and, in extreme cases, can turn linear algorithms into quadratic ones. \cite{tabulation,rust-quadratic,impala-quadratic,algorithm-attack}

HalftimeHash is a new ``universe collapsing'' hash family, designed to hash long strings into short ones. \cite{linear-hash-functions,hashing-without-primes-revisited,cuckoo-journal}
This differs from short-input families like SipHash or tabulation hashing, which are suitable for hashing short strings to a codomain of 64 bits. \cite{siphash,tabulation}
Universe collapsing families are especially useful for composition with short-input families: when $n$ long strings are to be handled by a hash-based algorithm, a universe-collapsing family that reduces them to hash values of length $c \lg n$ bits for some suitable $c > 2$ produces zero collisions with probability $1-O(n^{2-c})$.
A short-input hash family can then treat the hashed values as if they were the original input values. \cite{universe-collapse-linear-probing,siphash,tabulation,simple-hash-functions-work}
This technique applies not only to hash tables, but also to message-authentication codes, load balancing in distributed systems, privacy amplification, randomized geometric algorithms, Bloom filters, and randomness extractors. \cite{poly1305,chord,privacy-amplification,random-closest-pair,simple-hash-functions-work,fuzzy-extractors}

On strings longer than 1KB, HalftimeHash is typically 55\% faster than clhash, the AU hash family that comes closest in performance.

HalftimeHash also has tunable output length and low probabilities of collision for applications that require them, such as one-time authentication.\cite{nacl}
The codomain has size 16, 24, 32, or 40 bytes, and $\varepsilon$ varies depending on the codomain (see Figure \ref{frontier} and Section \ref{performance}).

\subsection{Portability}

In addition to high speed on long strings, HalftimeHash is designed for a simple implementation that is easily portable between programming languages and machine ISA's.
HalftimeHash uses less than 1200 lines of code in C++ and can take advantage of vector ISA extensions, including AVX-512, AVX2, SSE, and NEON.\footnote{\url{https://github.com/jbapple/HalftimeHash}}

Additionally, no multiplications from $\ints_{2^{64}} \times \ints_{2^{64}}$ to $\ints_{2^{128}}$ are needed.
This is in support of two portability goals -- the first is portability to platforms or programming languages without native widening unsigned 64-bit multiplications.
Languages like Java, Python, and Swift can do these long multiplications, but not without calling out to C or slipping into arbitrary-precision-integer code.
The other reason HalftimeHash avoids 64-bit multiplications is portability to SIMD ISA extensions, which generally do not contain widening 64-bit multiplication.
% TODO: check this

%% The x86-64 ISA extensions SSE2, AVX2, and AVX-512F all contain instructions to simultaneously multiply multiple pairs of 32-bit words, producing multiple 64-bit values.
%% Aarch64 has similar instructions in the NEON set.
% The POWER ISA also contains this, but only on 128 bits at once?

\subsection{Prior almost-universal families}

There are a number of fast hash algorithms that run at rates exceeding 8 bytes per cycle on modern x86-64 processors, including Fast Positive Hash, falk\-hash, xxh, Meow\-Hash, and UMASH, and cl\-hash. \cite{smhasher}
Of these, only cl\-hash and U\-MASH include claims of being AU; each of these uses finite fields and the x86-64 instruction for carryless (polynomial) multiplication.
% TODO: GHASH, Badger, Poly1305, ...
%Some similar previous work on AU hashing long strings is no longer as competitive, performance-wise, as it once was, including UHASH and VHASH. \cite{clhash,umac,vmac}

Rather than tree hashing, hash families like clhash and UMASH use polynomial hashing (based on Horner's method) to hash variable-length strings down to fixed-size output.
That approach requires 64-bit multiplication and also reduction modulo a prime (in $\ints$ or in $\ints_2[x]$), limiting its usability in SIMD ISA extensions.

\subsection{Outline}

The rest of this paper is organized as follows: Section \ref{prior-work} covers prior work that HalftimeHash builds upon.
Section \ref{gehc} introduces a new generalization of Nandi's ``Encode, Hash, Combine'' algorithm.\cite{ehc-nandi}
Section \ref{implementation} discusses specific implementation choices in HalftimeHash to increase performance.
Section \ref{performance} analyzes and tests HalftimeHash's performance.

\section{Notations and Conventions}

Input string length $n$ is measured in 32-bit {\em words}.
``32-bit multiplication'' means multiplying two unsigned 32-bit words and producing a single 64-bit word.
``64-bit multiplication'' similarly refers to the operation producing a 128-bit product.
All machine integers are unsigned.

Sequences are denoted by angled brackets: ``$\langle$'', ``$\rangle$'', and $\triangleleft$ prepends a character onto a string.
Subscripts indicate a numbered component of a sequence, starting at 0.
Contiguous half-open subsequences are denoted ''$x[y,z)$'', meaning $\langle x_y, x_{y+1}, \dots, x_{z-1} \rangle$.

$\bot$ is a new symbol not otherwise in the alphabet of words

$\varepsilon$ is called the {\em collision probability} of $H$; it is inversely related to $H$'s {\em output entropy}, $-\lg \varepsilon$.
The seed is sometimes referred as {\em input entropy}, which is distinguished from the output entropy both because it is an explicit part of the input and because it is measured in words or bytes, not bits.

Each step of HalftimeHash applies various transforms to groups of input values.
These groups are called {\em instances}.
The processing of a transform on a single instance is called an {\em execution}.

Instances are logically contiguous but physically strided, for the purpose of simplifying SIMD processing.
A physically contiguous set between two items in a single instance is called a {\em block}; the number of words in a block is called the {\em block size}.
Because instances are logically contiguous, when possible, the analysis will elide references to the block size.

Tree hashing examples use a hash family parameter $H$ that takes two words as input, but this can be easily extended to hash functions taking more than two words of input, in much the same way that binary trees are a special case of B-trees.

HalftimeHash produces output that is collision resistant among strings of the same length.
Adding collision resistance between strings of {\em different} lengths to such a hash family requires only appending the length at the end of the output.
%This turns, for instance, a hash family that produces 24 bytes of output into a hash family that produces 32 bytes of output.

%The main portion of the text describes a particular instantiation of HalftimeHash that produces 24 bytes of output.
%This will be generalized in Section~\ref{performance}.
Variants will be specified by their number of output bytes: HalftimeHash16, HalftimeHash24, HalftimeHash32, or HalftimeHash40.
%HalftimeHash24 has $\varepsilon < 2^{-83}$, as discussed in Section~\ref{performance}.

Except where otherwise mentioned, all benchmarks were run on an Intel i7-7800x (a Skylake chip that supports AVX512), running Ubuntu 18.04, with clang++ 11.0.1.

\section{Prior work}
\label{prior-work}

%% Roughly, each execution HalftimeHash can be thought of as a tree with a hashing primative ``NH'' at the nodes and instances of ``EHC'' at the leaves.
%% The prior work on each of these components is described in this section.

This section reviews hashing constructions that form components of HalftimeHash.
In order to put these in context, a broad outline of HalftimeHash is in order.

HalftimeHash can be thought of as a tree-based, recursively-defined hash function.
The leaves of the tree are the words of the unhashed input; the root is the output value.
Every internal node has multiple inputs and a single output, corresponding with the child and parent nodes in the tree.

To a first approximation, a string is hashed by breaking it up into some number of contiguous parts, hashing each part, then combining those hash values.
When the size of the input is low enough, rather than recurse, a construction called ``Encode, Hash, Combine'' (or ``EHC'') is used to hash the input.

%% The overal structure of HalftimeHash is from a tree-like hash invented by Carter and Wegman;
%% the internal nodes use NH, a hashing primitive from UHASH;
%% the EHC at the leaves is described in this section (and generalized for use in HalftimeHash in Section~\ref{gehc}).

\subsection{Tree hash}

HalftimeHash's structure is based on a tree-like hash as described by Carter and Wegman. \cite[Section 3]{carter-wegman-79}
To hash a string, we use $\lceil \lg n \rceil$ randomly-selected keys $k_i$ and a hash family $H$ that hashes two words down to one.
%Now for any $n \in \nats$ let $\lfloor\!\lfloor n \rfloor\!\rfloor$ denote the largest power of $2$ that is less than $n$.
Then the tree hash $T$ of a string $s[0,n)$ is defined recursively as:
\begin{equation}
\label{algebraic-badger}
\begin{array}{l}
T(k, \langle x \rangle) \defeq x \\
T(k, s[0,n)) \defeq H
(k_{  \lceil \lg n - 1 \rceil},
T(k, s[0, 2^{\lceil \lg n - 1 \rceil}) ),
T(k, s[ 2^{\lceil \lg n - 1 \rceil}, n) ))
\end{array}
\end{equation}

Carter and Wegman show that if $H$ is $\varepsilon$-AU, $T$ is $m\varepsilon$-AU for input that has length exactly $2^m$.
Later, Boesgaard et al. extended this proof to strings with lengths that are not a power of two.\cite{badger}

\subsection{NH}

In HalftimeHash, NH, an almost-universal hash family, is used at the nodes of tree hash to hash small, fixed-length sequences:\cite{umac}
\[
\sum_{i=0}^m (d_{2i} + s_{2i})(d_{2i+1} + s_{2i+1})
\]
where $d, s \in \ints_{2^{32}}^{2m+2}$ are the input string and the input entropy, respectively.
The $d_j + s_j$ additions are in the ring $\ints_{2^{32}}$, while all other operations are in the ring $\ints_{2^{64}}$.
NH is $2^{-32}$-AU.
In fact, it satisfies a stronger property, $2^{-32}$-A$\Delta$U:\cite{umac}

\begin{definition}
  A hash family $H$ is said to be {\em $\varepsilon$-almost $\Delta$-universal} (or just A$\Delta$U) when
  \[
  \forall x,y,\delta, \Pr_s[H(s,x) - H(s,y) = \delta] \leq \varepsilon \in o(1)
  \]
\end{definition}

In tree nodes (though not in EHC, covered below), a variant of NH is used in which the last input pair is not hashed, thereby increasing performance:
\[
\left(\sum_{i=0}^{m-1} (d_{2i} + s_{2i})(d_{2i+1} + s_{2i+1})\right) + d_{2m} + 2^{32} d_{2m+1}
\]
This hash family is still $2^{-32}$-AU.\cite{badger}

\subsection{Encode, Hash, Combine}

At the leaves of the tree hash, HalftimeHash uses the ``Encode, Hash, Combine'' algorithm.\cite{ehc-nandi}
EHC is parameterized by an erasure code with ``minimum distance'' $k$, which is a map on sequences of words such that any two input values that differ in {\em any} location produce encoded outputs that differ in {\em at least $k > 1$} locations after encoding.

The EHC algorithm is:
\begin{enumerate}
\item A sequence of words is processed by an erasure code with minimum distance $k$, producing a longer encoded sequence.
\item Each word in the encoded sequence is hashed using an A$\Delta$U family with independently and randomly chosen input entropy.
\item A linear transformation $T$ is applied to the resulting sequence of hash values.
  The codomain of $T$ has dimension $k$, and $T$ must have the property that any $k$ columns of it are linearly independent.
\end{enumerate}

Nandi proved that if the EHC matrix product is over a finite field, EHC is $\varepsilon^k$-AU.
This AU collision probability could be achieved on the same input by instead running $k$ copies of NH, but that would perform $m k$ multiplications to hash $m$ words, while EHC requires $m + k$ multiplications, excluding the multiplications implicit in applying $T$.
That exclusion is the topic of Section~\ref{gehc}.

\section{Generalized EHC}
\label{gehc}

At first glance, EHC might not look like it will reduce the number of multiplications needed, as the application of linear transformations usually requires multiplication.
However, since $T$ is not part of the randomness of the hash family, it can be designed to contain only values that are trivial to multiply by, such as powers of $2$.

The constraint in \cite{ehc-nandi} requires that any $k$ columns of $T$ form an invertible matrix.
This is not feasible in linear transformations on $\ints_{2^{64}}$ in most useful dimensions.
For instance, in HalftimeHash24, a $3 \times 9$ matrix $T$ is used.
Any such matrix will have at least one set of three columns with an even determinant, and which therefore has a non-trivial kernel.

\begin{proof}
  Let $U$ be a matrix over $\ints_2$ formed by reducing each entry of $T$ modulo $2$.
  Then $(\det T) \bmod 2 \equiv \det U$.
  Since there are only 7 unique non-zero columns of size 3 over $\ints_2$, by the pigeonhole principle, some two columns $x, y$ of $U$ must be equal.
  Any set of columns that includes both $x$ and $y$ has a determinant of $0 \bmod 2$.
  \qed
\end{proof}

%Furthermore, the input to each section is independent? Partially independent, based on the loss at the matrix multiplication step? Since each partition

Let $k$ be the minimum distance of the erasure code.
While Nandi proved that EHC is $\varepsilon^k$-AU over a finite field, $\ints_{2^{64}}$ is not a finite field.
However, there are similarities to a finite field, in that there are some elements in $\ints_{2^{64}}$ with inverses.
Some other elements in $\ints_{2^{64}}$ are zero divisors, but only have one value that they can be multiplied by to produce 0.
A variant of Nandi's proof is presented here as a warm-up to explain the similarities. \cite{ehc-nandi}

\begin{lemma}
  When the matrix product is taken over a field, if the hash function $H$ used in step 2 is $\varepsilon$-A$\Delta$U, EHC is $\varepsilon^k$-A$\Delta$U.
\end{lemma}
\begin{proof}
  Let $\bar{H}$ be defined as $\bar{H}(s, x)_i \defeq H(s_i, x_i)$.
  Let $J$ be the encoding function that acts on $x$ and $y$, producing an encoding of length $e$.
  Given that $x$ and $y$ differ, let $F$ be $k$ locations where $J(x)_i \neq J(y)_i$.
  Let $T|_F$ be the matrix formed by the columns of $T$ where the column index is in $F$ and let $\bar{H}|_F$ similarly be $\bar{H}$ restricted to the indices in $F$.
  Conditioning over the $e - k$ indices not in $F$, we want to bound
  \begin{equation}
    \label{ehc-delta}
    \Pr_s[T|_F \bar{H}|_F(s, J(x)) - T|_F \bar{H}|_F(s, J(y)) = \delta]
  \end{equation}
  Since any $k$ columns of $T$ are independent, $T|_F$ is non-singular, and the equation is equivalent to $\bar{H}|_F(s, J(x)) - \bar{H}|_F(s, J(y)) = {T|_F}^{-1} \delta$, which implies
  \[
  \bigwedge_{i \in F} H(s_i, J(x)_i) - H(s_i, J(y)_i) = \beta_i
  \]
  where $\beta \defeq {T|_F}^{-1} \delta$.

  Since the $s_i$ are all chosen independently, the probability of the conjunction is the product of the probabilities, showing
  \[
  \begin{array}{rl}
    &  \Pr_s[T|_F \bar{H}|_F (s,J(x)) - T|_F \bar{H}|_F(s,J(y)) = \delta] \\
  \leq &  \prod_{i \in F} Pr_s[H(s_i, J(x)_i) - H(s_i, J(y)_i) = \beta_i]
  \end{array}
  \]
  and since $H$ is A$\Delta$U, this probability is $\varepsilon^k$.  \qed
\end{proof}

Note that this lemma depends on $k$ being the minimum distance of the code.
If the distance were less than $k$, then the matrix would be smaller, increasing the probability of collisions.

In the non-field ring $\ints_{2^{64}}$, the situation is altered.
``Good'' matrices are those in which the determinant of any $k$ columns is divisible only by a small power of two.
The intuition is that, since matrices in $\ints_{2^{64}}$ with odd determinants are invertible, the ``closer'' a determinant is to odd (meaning it is not divisible by large powers of two), the ``closer'' it is to invertible.

\begin{theorem}
  Let $p$ be the largest power of 2 that divides the determinant of any $k$ columns in $T$.
  The EHC step of HalftimeHash is $2^{k(p-32)}$-A$\Delta$U when using NH as the hash family.
\end{theorem}

\begin{proof}
  In HalftimeHash, the proof of the lemma above unravels at the reliance upon the trivial kernel of $T|_F$.
  The columns of $T$ in HalftimeHash are linearly independent, so the matrix $T|_F$ is injective in rings without zero dividers, but not necessarily injective in $\ints_{2^{64}}$.

  However, even in $\ints_{2^{64}}$, the adjugate matrix $\adj(A)$ has the property that $A \cdot \adj(A) = \adj(A) \cdot A = \det(A) I$.
  Let $\det(T|_F) = q2^{p'}$, where $q$ is odd and $p' \le p$.
  Now (\ref{ehc-delta}) reduces to
  \[
  \begin{array}{rl}
    &   \Pr_s[T|_F \bar{H}|_F(s,x) - T|_F \bar{H}|_F(s,y) = \delta]\\
    \leq &  \Pr_s[\adj(T|_F) T|_F \bar{H}|_F(s,x) - \adj(T|_F) T|_F \bar{H}|_F(s,y) = \adj(T|_F) \delta] \\
    = &  \Pr_s[q2^{p'}\bar{H}|_F(s,x) - q2^{p'}\bar{H}|_F(s,y) = \adj(T|_F) \delta] \\
    = &  \Pr_s[2^{p'}\bar{H}|_F(s,x) - 2^{p'}\bar{H}|_F(s,y) = q^{-1} \adj(T|_F) \delta]
  \end{array}
  \]

  Now letting $\beta = q^{-1} \adj(T|_F) \delta$ and letting the modulo operator extend pointwise to vectors, we have

  \[
  \begin{array}{rl}
    = &  \Pr_s[\bar{H}|_F(s,x) - \bar{H}|_F(s,y) \equiv \beta \bmod 2^{64-p'}] \\
    = &  \Pr_s\left[\bigwedge_{i \in F} H(s_i,x_i) - H(s_i,y_i) \equiv \beta_i \bmod 2^{64-p'}\right] \\
    = &  \prod_{i \in F} \Pr_s\left[ H(s_i,x_i) - H(s_i,y_i) \equiv \beta_i \bmod 2^{64-p'}\right] \\
    = & \left(2^{p'} 2^{-32}\right)^{|F|} = 2^{k(p'-32)}
  \end{array}
  \]
  This quantity is highest when $p'$ is at its maximum over all potential sets of columns $F$, and $p'$ is at most $p$, by the definition of $p$.
  \qed
  % TODO: this can be reduced by shuffling the column order
\end{proof}

This generalized version of EHC is used in the implementation of HalftimeHash described in Section~\ref{implementation}, with $p \le 2^3$.

\section{Implementation}
\label{implementation}

This section describes the specific implementation choices made in HalftimeHash to ensure high output entropy and high performance.
The algorithm performs the following steps:

\begin{itemize}
\item Generalized EHC on instances of the unhashed input, producing 2, 3, 4, or 5 output words (of 64 bits each) per input instance
\item 2, 3, 4, or 5 exexutions of tree hash (with independently and randomly chosen input entropy) on the output of EHC, with NH at each internal node, producing a sequence of words logarithmic in the length of the input string, as described below in Equation~\ref{stack-construction}
\item NH on the output of each tree hash, producing 16, 24, 32, or 40 bytes
\end{itemize}

\subsection{EHC}

In addition to the trivial distance-2 erasure code of XOR'ing the words together and appending that as an additional word, HalftimeHash uses non-linear erasure codes discovered by Gab\-ri\-el\-yan with minimum distance 3, 4, or 5. \cite{9-7-erasure-code,10-7-erasure-code,9-5-erasure-code}

For the linear transformations, HalftimeHash uses matrices $T$ selected so that the largest power of 2 that divides any determinant is $2^2$ or $2^3$.
For instance, for the HalftimeHash24 variant, $T$ has a $p$ of $2^2$:

\begin{displaymath}
  \left(
\begin{array}{rrrrrrrrr}
  0 & 0 & 1 & 4 & 1 & 1 & 2 & 2 & 1\\
  1 & 1 & 0 & 0 & 1 & 4 & 1 & 2 & 2\\
  1 & 4 & 1 & 1 & 0 & 0 & 2 & 1 & 2
\end{array}
\right)
\end{displaymath}

For other output widths, HalftimeHash uses

\[
\begin{tabular}{|r|c|c|c|}
  \hline  & HalftimeHash16 & HalftimeHash32 & HalftimeHash40 \\
  \hline $T$ &
$\left(
\begin{array}{rrrrrrrrrrrr}
  1 & 0 & 1 & 1 & 2 & 1 & 4\\
  0 & 1 & 1 & 2 & 1 & 4 & 1
\end{array}
\right)$
&
$\left(
\begin{array}{rrrrrrrrrr}
 0 & 0 & 0 & 1 & 1 & 4 & 2 & 4 & 1 & 1 \\
 0 & 1 & 2 & 0 & 0 & 1 & 1 & 2 & 4 & 1 \\
 2 & 0 & 1 & 0 & 4 & 0 & 1 & 1 & 1 & 1 \\
 1 & 1 & 0 & 1 & 0 & 0 & 4 & 1 & 2 & 8
\end{array}
\right)$
&
$\left(
\begin{array}{rrrrrrrrr}
 1 & 0 & 0 & 0 & 0 & 1 & 1 & 2 & 4\\
 0 & 1 & 0 & 0 & 0 & 1 & 2 & 1 & 7\\
 0 & 0 & 1 & 0 & 0 & 1 & 3 & 8 & 5\\
 0 & 0 & 0 & 1 & 0 & 1 & 4 & 9 & 8\\
 0 & 0 & 0 & 0 & 1 & 1 & 5 & 3 & 9
\end{array}
\right)$ \\
\hline $p$ & $2^2$ & $2^3$ & $2^3$ \\
\hline
\end{tabular}
\]

The input group lengths for the EHC input are 6, 7, 7, and 5, as can be seen from the dimensions of the matrices: $\text{columns} + 1 - \text{rows}$.
Note that each of these matrices contains coefficients that can be multiplied by with no more than two shifts and one addition.

\subsection{Tree hash}

For the tree hashing at internal nodes (above the leaf nodes, which use EHC), $k \in \{2, 3, 4, 5\}$ tree hashes are executed with independently-chosen input entropy, producing output entropy of $-k \lg \varepsilon$.
From the result from Carter and Wegman on the entropy of tree hash of a tree of height $m$, the resulting hash function is $m\varepsilon^k$-AU.

The key lemma they need is that almost universality is composable:

\begin{lemma}[Carter and Wegman]
  If $F$ is $\varepsilon_F$-AU, $G$ is $\varepsilon_G$-AU, then
  \begin{itemize}
  \item $F \circ G$ where $F \circ G (\langle k_F, k_G \rangle, x) \defeq F(k_F,G(k_G, x))$ is $(\varepsilon_F + \varepsilon_G)$-AU.
  \item $\langle F, G\rangle$ where $\langle F, G \rangle(\langle k_F, k_G \rangle, \langle x, y \rangle) \defeq \langle F(k_F,x), G(k_G,y) \rangle$, is $\text{max}(\varepsilon_F, \varepsilon_G)$-AU, even if $F=G$ and $k_F = k_G$.
%  \item $F \circ \langle G,H \rangle$ where $(F \circ \langle G,H \rangle) (\langle k_F, k_G, k_H \rangle , \langle p,q \rangle) \defeq F(k_F, G(k_G,p), H(k_H, q))$ is a family parameterized by the combination of keys for $F$ , $G$, and $H$, even if $G = H$ and $k_G = k_H$.
  \end{itemize}
\end{lemma}

The approach in Badger of using Equation~\ref{algebraic-badger} to handle words that are not in perfect trees can be increased in speed with the following method:
For HalftimeHash, define $\widehat{T}$ as a family taking as input sequences of any length $n$ and producing sequences of length $\lceil \lg n \rceil$ as follows, using Carter and Wegman's $T$ defined in Section~\ref{prior-work}:

\begin{equation}
\label{stack-construction}
\begin{array}{l}
\widehat{T}_0(k, \langle \rangle) \defeq \langle \bot \rangle \\
\widehat{T}_0(k, \langle x \rangle) \defeq \langle x \rangle \\
\widehat{T}_{i+1}(k, s[0,n)) \defeq \left\{
  \begin{array}{rcll}
   \bot &\triangleleft& \widehat{T}_i(k[1,\lceil\lg n\rceil), s[0,n)) & \text{if } 2^i > n \\
    T(k, s[0, 2^i))) &\triangleleft& \widehat{T}_{i}(k[1,\lceil\lg n\rceil),s[2^i, n)) & \text{if } 2^i \le n
  \end{array}
  \right.
\end{array}
\end{equation}

There is one execution of $T$ for every 1 in the binary representation of $n$.
%This makes $\widehat{T}$ produce a sequence with the same number of elements as there are levels in the execution of $T$.
By an induction on $\lceil \lg n \rceil$ using the composition lemma, $\widehat{T}$ is $\varepsilon \lceil \lg n \rceil$-AU.

The output of $\widehat{T}$ is then hashed using an NH instance of size $\lceil \lg n \rceil$.
This differs from Badger, where $T$ is used to fully consume the input without the use of additional input entropy;
$T$ produces a single word per execution, while $\widehat{T}$ needs to be paired with NH post-processing in order to achieve that.\cite{badger}
Empirically, $\widehat{T}$ has better performance than the Badger approach.

\section{Performance}
\label{performance}

This section tests and analyzes HalftimeHash performance, including an analysis of the output entropy.

\subsection{Analysis}
%This section presents metrics of an execution of HalftimeHash, including the collision probability.
%, number of multiplications performed, and the amount of input entropy used.
%This section will treat HalftimeHash as abstract, rather than focusing on a single version with 24 bytes of output, as described above.
The parameters used in this analysis are:

\begin {description}
\item[$b$] the number of 64-bit words in a block.
  Blocks are used to take advantage of SIMD units.
\item[$d$] is the number of elements in each EHC instance before applying the encoding.
\item[$e$] is the number of blocks in EHC after applying the encoding.
\item[$f$] is fanout, the width of the NH instance at tree hash nodes.
\item[$k$] is the number of blocks produced by the Combine step of EHC.
  This is also the minimum distance of the erasure code, as described above.
\item[$p$] is the maximum power of 2 that divides a determinant of any $k \times k$ matrix made from columns of the matrix $T$; doubling $p$ increases $\varepsilon$ by a factor of $2^k$.
\item[$w$] is the number of blocks in each item used in the Encode step of EHC.
\end{description}

In HalftimeHash24, \[(w, d, e, k, b, f, p) = (3, 7, 9, 3, 8, 8, 2^2)\]

Each EHC execution reads in $d w$ blocks, produces $e$ blocks, uses $e w$ words of input entropy, and performs $e w$ multiplications.

For the tree hash portion of HalftimeHash, the height of the $k$ trees drives multiple metrics.
Each tree has $\lfloor n / b d w \rfloor$ blocks as input and every level execution forms a complete $f$-ary execution tree.
The height of the tree is thus $h \defeq \left\lfloor \log_f \lfloor n / b d w \rfloor \right\rfloor$.

\begin{lemma}
The tree hash is $2^{ k\lg h - 32k}$-AU.
\end{lemma}
\begin{proof}
  Carter and Wegman showed that tree hash has collision probability of $h \varepsilon$, where $\varepsilon$ is the collision probability of a single node.
  Each tree node uses NH, so a single tree has collision probability $h 2^{-32}$.
  A collision occurs for HalftimeHash at the tree hash stage if and only if all $k$ trees collide, which has probability $\left(h 2^{-32}\right)^k$, assuming that the EHC step didn't already induce a collision. \qed
\end{proof}

The amount of input entropy needed is proportional to the height of the tree, with $f - 1$ words needed for every level.
HalftimeHash uses different input entropy for the $k$ different trees, so the total number of 64-bit words of input entropy used in the tree hash step is $(f - 1) h k$.

The number of multiplications performed is identical to the number of input words, $k b \lfloor n / b d w \rfloor$.

%% \begin{tabular}{|r|c|c|}
%%   \hline & {\bf EHC} & {\bf Tree hash}\\
%%   \hline {\bf Multiplications (each node)} & $b e w$ & $b (f-1)$ \\
%%   \hline {\bf Multiplications (total)} & $b e w \lfloor n / b d w\rfloor$ & $k b \lfloor n / b d w \rfloor$ \\
%%   \hline {\bf In Entropy (each tree $\times$ level)} & N/A & $f-1$ \\
%%   \hline {\bf In Entropy (total)} & $e w$ & $k (f-1) \left\lfloor \log_f \lfloor n / b d w \rfloor \right\rfloor$ \\
%%   \hline {\bf Out Entropy (total)} & $k (32-p)$ & $32k - k\lg\left\lfloor\log_f \lfloor n/b d w\rfloor\right\rfloor$\\
%%   \hline {\bf Output words (total)} & $k b \lfloor n / b d w\rfloor $ & $b f k \left\lfloor \log_f \lfloor n / b d w \rfloor \right\rfloor$\\
%%   \hline
%% \end{tabular}

The result of the tree hash is processed through NH, which uses $b f h k$ words of entropy and just as many multiplications.

There can also be as much as $b d w$ words of data in the raw input that are not read by HalftimeHash, as they are less than the input size of one instance of EHC.
Again, NH is used on this data, but now hashing $k$ times, since this data has not gone through EHC.
That requires $b d w k$ words of entropy and just as many multiplications.

For this previously-unread data, the number of words of entropy needed can be reduced by nearly a factor of $k$ using the Toeplitz construction.
Let $r$ be the sequence of random words used to hash it.
Instead of using $r[i b d w, (i+1)b d w)$ as the keys to hash component $i$ with, HalftimeHash uses $r[i, b d w + i)$.
This construction for multi-part hash output is A$\Delta$U. \cite{ehc-nandi,woelfel-toeplitz}

\subsection{Cumulative analysis}

The combined collision probability is
$2^{-32k}\left(2^{kp} + h^k + 1\right)$.
For HalftimeHash24, and for strings less than an exabyte in length, this is more than 83 bits of entropy.

The combined input entropy needed (in words) is
$
e w
+ (f-1) h k
+ b f h k
+ b d w + k - 1
$
HalftimeHash24 requires 8.4KB input entropy for strings of length up to one megabyte and 34KB entropy for strings of length up to one exabyte.

The number of multiplications is dominated by the EHC step, since the total is $(e w + k) b \lfloor n / b d w \rfloor + O(\log n)$ and $e w$ is significantly larger than $k$.
For a string of length 1MB, 84\% of the multiplications happen in the EHC step. %, and the number of multiplications is about one per ten bytes of input.
Intel's VTune tool show the same thing: 86\% of the clock cycles are spent in the EHC step.
Similarly, clhash and UMASH, which are based on 64-bit carryless NH, have their execution times dominated by the multiplications in their base step.~\cite{clhash,umash}

\subsection{Benchmarks}
\label{benchmarks}
%% This section covers performance testing for HalftimeHash, especially compared to clhash and UMASH, the two fastest AU families on long strings.
%% Each of those are based on NH over $\ints_2[x]$, rather than $\ints_{2^{64}}$.

HalftimeHash passes all correctness and randomness tests in the SMHasher test suite; for a performance comparison, see Figure \ref{smhasher-speed} and \cite{smhasher}.

\begin{figure}
  \includegraphics[width=\textwidth]{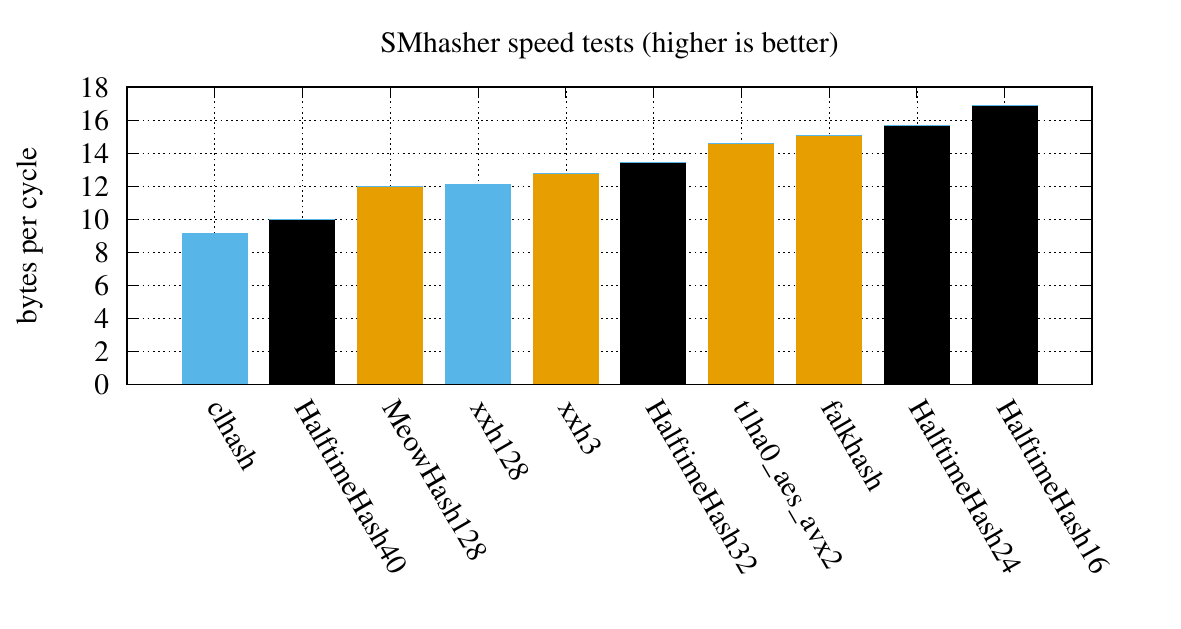}
\caption{
  \label{smhasher-speed}
    The two fastest variants of HalftimeHash are faster than all hash families in the SMHasher suite on 256KiB strings on an i7-7800x, even families that come with no AU guarantees. \protect\cite{smhasher}
    Of the families here, only HalftimeHash and xxh128 pass all SMHasher tests, and only HalftimeHash and clhash are AU.
}
\end{figure}

Figure~\ref{frontier} displays the relationship between output entropy and throughput for HalftimeHash, UMASH, and clhash.\footnote{UMASH and clhash are the fastest AU families for string hashing}
Adding more output entropy increases the number of non-linear arithmetic operations that any hash function has to perform.\cite{ehc-nandi}
%Nandi showed that this is true in the general case, as there is a matching upper and lower bound for the number of non-linear operations to be performed for a certain hash output width.
The avoidance of doubling the number of multiplications for twice the output size is one of the primary reasons that HalftimeHash24, -32, and -40 are faster than running clhash or UMASH with 128-bit output.
(The other is that carryless multiplication is not supported as a SIMD instruction.)

\begin{figure}
\includegraphics[width=\textwidth]{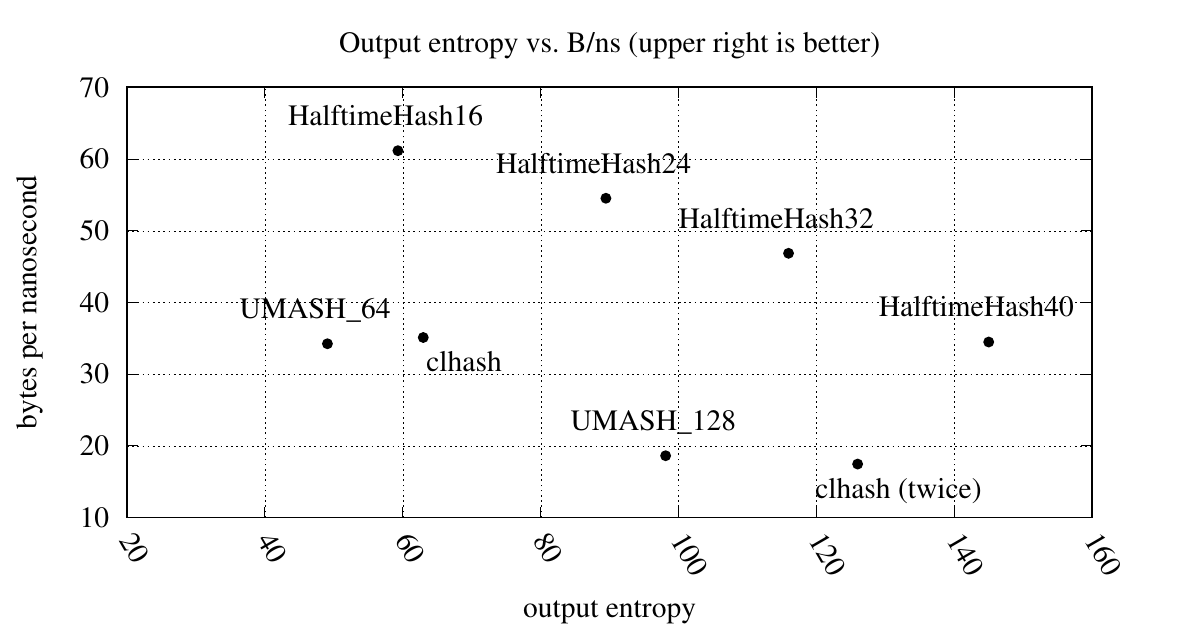}
\caption{
  \label{frontier}
  Trade-offs for almost-universal string hashing functions on strings of size 250KB on an i7-7800x.
  UMASH comes in two variants based on the output width in bits; clhash doesn't, but running clhash twice is included in the chart.
  For each clhash / UMASH version, at least one version of HalftimeHash is faster and has lower collision probability. \protect\cite{layer-of-maxima}
}
\end{figure}

Figure \ref{vs-cl} adds comparisons between clhash, UMASH, and HalftimeHash across input sizes and processor manufacturers.
Although these two machines support different ISA vector extensions, the pattern is similar: for large enough input, HalftimeHash's throughput exceeds that of the carryless multiplication families.
\begin{figure*}
\begin{tabular}{cc}
\includegraphics[width=6.0cm]{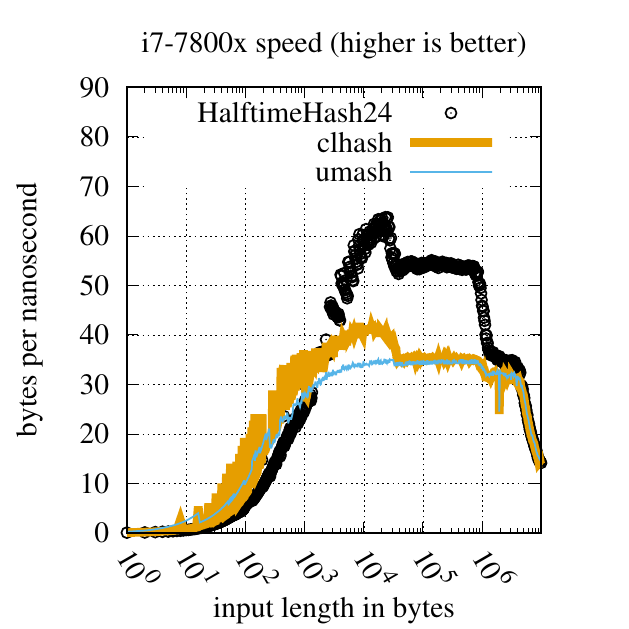}
&
\includegraphics[width=6.0cm]{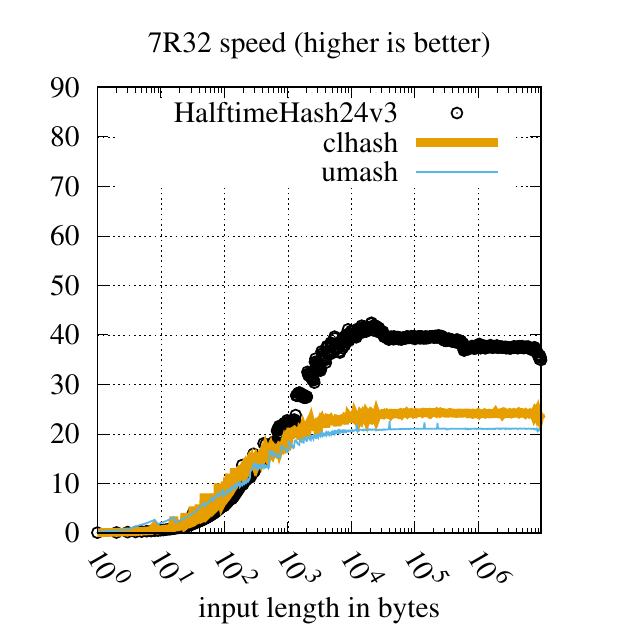}
\end{tabular}
\caption{
  \label{vs-cl}
  Comparison of Intel (i7-7800x) and AMD (EC2 c5a.large, 7R32) performance.
  On both chips HalftimeHash24 is faster than clhash and UMASH for long strings.
  The ``v3'' after the name of the AMD HalftimeHash indicates block size: v3 means a 256-bit block size, while v4 (the default) means 512-bit block size.
  AMD chips do not support AVX-512, but still HalftimeHash with 256-bit blocks exceeds the speed of clmul-based hashing methods by up to a factor of 2.
}
\end{figure*}

\section{Future work}

Areas of future research include:

\begin{itemize}
\item Combining HalftimeHash, which is designed for long input, with a fast family for short input
\item Tuning for JavaScript, which has no native 32-bit integer support
\item Comparisons against hash algorithms in the Linux kernel, including Poly1305 and \texttt{crc32\_pclmul\_le\_16}
\item Benchmarks on POWER and ARM ISA's
\item EHC benchmarks using 64-bit multiplication -- carryless or integral
\end{itemize}

\section*{Acknowledgments}
Thanks to Daniel Lemire, Paul Khuong, and Guy Even for helpful discussions and feedback.

\bibliographystyle{splncs04}
\bibliography{halftime-hash}

\begin{thebibliography}{10}
\providecommand{\url}[1]{\texttt{#1}}
\providecommand{\urlprefix}{URL }
\providecommand{\doi}[1]{https://doi.org/#1}

\bibitem{linear-hash-functions}
Alon, N., Dietzfelbinger, M., Miltersen, P.B., Petrank, E., Tardos, G.: Linear
  hash functions. J. ACM  \textbf{46}(5),  667--683 (Sep 1999),
  \url{https://doi.org/10.1145/324133.324179},
  \url{http://www.math.tau.ac.il/~nogaa/PDFS/linhash13.pdf}
  \url{https://tidsskrift.dk/brics/article/view/1880}

\bibitem{impala-quadratic}
Apple, J.: {Ensure monotonic count(distinct x) performance} (2015),
  \url{https://issues.apache.org/jira/browse/IMPALA-2653}, accessed on
  2020-12-32

\bibitem{siphash}
Aumasson, J.P., Bernstein, D.J.: {SipHash}: a fast short-input {PRF}. In:
  International Conference on Cryptology in India. pp. 489--508. Springer
  (2012)

\bibitem{poly1305}
Bernstein, D.J.: The {Poly1305-AES} message-authentication code. In:
  International Workshop on Fast Software Encryption. pp. 32--49. Springer
  (2005)

\bibitem{nacl}
Bernstein, D.J.: Cryptography in {NaCl}. Networking and Cryptography library
  \textbf{3}, ~385 (2009)

\bibitem{umac}
Black, J., Halevi, S., Krawczyk, H., Krovetz, T., Rogaway, P.: {UMAC}: Fast and
  secure message authentication. In: Annual International Cryptology
  Conference. pp. 216--233. Springer (1999),
  \url{https://link.springer.com/chapter/10.1007/3-540-48405-1_14},
  \url{https://web.cs.ucdavis.edu/~rogaway/papers/umac-full.pdf}

\bibitem{badger}
Boesgaard, M., Scavenius, O., Pedersen, T., Christensen, T., Zenner, E.: Badger
  - a fast and provably secure {MAC}. Cryptology ePrint Archive, Report
  2004/319 (2004), \url{https://eprint.iacr.org/2004/319}

\bibitem{layer-of-maxima}
Buchsbaum, A.L., Goodrich, M.T.: Three-dimensional layers of maxima. In:
  European Symposium on Algorithms. pp. 257--269. Springer (2002)

\bibitem{carter-wegman-79}
Carter, J.L., Wegman, M.N.: Universal classes of hash functions. Journal of
  Computer and System Sciences  \textbf{18}(2),  143 -- 154 (1979),
  \url{http://www.sciencedirect.com/science/article/pii/0022000079900448}

\bibitem{simple-hash-functions-work}
Chung, K.M., Mitzenmacher, M., Vadhan, S.: Why simple hash functions work:
  Exploiting the entropy in a data stream. Theory of Computing  \textbf{9}(30),
   897--945 (2013). \doi{10.4086/toc.2013.v009a030},
  \url{http://www.theoryofcomputing.org/articles/v009a030}

\bibitem{algorithm-attack}
Crosby, S.A., Wallach, D.S.: Denial of service via algorithmic complexity
  attacks. In: Proceedings of the 12th {USENIX} Security Symposium. pp. 29--44
  (2003)

\bibitem{hashing-without-primes-revisited}
Dietzfelbinger, M.: Universal Hashing via Integer Arithmetic Without Primes,
  Revisited, pp. 257--279. Springer International Publishing (2018)

\bibitem{random-closest-pair}
Dietzfelbinger, M., Hagerup, T., Katajainen, J., Penttonen, M.: A reliable
  randomized algorithm for the closest-pair problem. Journal of Algorithms
  \textbf{25}(1),  19--51 (1997)

\bibitem{fuzzy-extractors}
Dodis, Y., Ostrovsky, R., Reyzin, L., Smith, A.: Fuzzy extractors: How to
  generate strong keys from biometrics and other noisy data. SIAM Journal on
  Computing  \textbf{38}(1),  97–139 (Jan 2008). \doi{10.1137/060651380},
  \url{http://dx.doi.org/10.1137/060651380}

\bibitem{10-7-erasure-code}
Gabrielyan, E.: Erasure resilient (10,7) code (2005),
  \url{https://docs.switzernet.com/people/emin-gabrielyan/051102-erasure-10-7-resilient/},
  accessed: 2020-11-26

\bibitem{9-5-erasure-code}
Gabrielyan, E.: Erasure resilient {MDS} code with four redundant packets
  (2005),
  \url{https://docs.switzernet.com/people/emin-gabrielyan/051103-erasure-9-5-resilient/},
  accessed: 2020-11-26

\bibitem{9-7-erasure-code}
Gabrielyan, E.: Erausre resulient (9,7)-code (2005),
  \url{https://docs.switzernet.com/people/emin-gabrielyan/051101-erasure-9-7-resilient/},
  accessed: 2020-11-26

\bibitem{umash}
Khuong, P.: {UMASH}: a fast and universal enough hash (2020),
  \url{https://engineering.backtrace.io/2020-08-24-umash-fast-enough-almost-universal-fingerprinting/}

\bibitem{rust-quadratic}
Landau, J.: {Exposure of {HashMap} iteration order allows for $O(n^2)$ blowup}
  (2016), \url{https://github.com/rust-lang/rust/issues/36481}, accessed on
  2020-12-32

\bibitem{clhash}
Lemire, D., Kaser, O.: Faster 64-bit universal hashing using carry-less
  multiplications. Journal of Cryptographic Engineering  \textbf{6}(3),
  171--185 (2016)

\bibitem{ehc-nandi}
Nandi, M.: On the minimum number of multiplications necessary for universal
  hash functions. In: Cid, C., Rechberger, C. (eds.) Fast Software Encryption.
  pp. 489--508. Springer Berlin Heidelberg, Berlin, Heidelberg (2015)

\bibitem{cuckoo-journal}
Pagh, R., Rodler, F.F.: Cuckoo hashing. Journal of Algorithms  \textbf{51}(2),
  122 -- 144 (2004),
  \url{http://www.sciencedirect.com/science/article/pii/S0196677403001925},
  \url{https://www.itu.dk/people/pagh/papers/cuckoo-jour.pdf}

\bibitem{privacy-amplification}
Renner, R., K{\"o}nig, R.: Universally composable privacy amplification against
  quantum adversaries. In: Theory of Cryptography Conference. pp. 407--425.
  Springer (2005)

\bibitem{chord}
Stoica, I., Morris, R., Karger, D., Kaashoek, M.F., Balakrishnan, H.: Chord: A
  scalable peer-to-peer lookup service for internet applications. ACM SIGCOMM
  Computer Communication Review  \textbf{31}(4),  149--160 (2001)

\bibitem{universe-collapse-linear-probing}
Thorup, M.: String Hashing for Linear Probing, pp. 655--664 (2009).
  \doi{10.1137/1.9781611973068.72},
  \url{https://epubs.siam.org/doi/abs/10.1137/1.9781611973068.72}

\bibitem{tabulation}
Thorup, M.: Fast and powerful hashing using tabulation. CoRR
  \textbf{abs/1505.01523} (2017), \url{http://arxiv.org/abs/1505.01523}

\bibitem{smhasher}
Urban, R., et~al.: Smhasher (2020), \url{https://github.com/rurban/smhasher}

\bibitem{woelfel-toeplitz}
Woelfel, P.: Efficient strongly universal and optimally universal hashing. In:
  Mathematical Foundations of Computer Science 1999, 24th International
  Symposium. LNCS, vol.~1672, pp. 262--272 (1999),
  \url{https://pages.cpsc.ucalgary.ca/~woelfel/paper/efficient_strongly/efficient_strongly.pdf}

\end{thebibliography}

\end{document}